\documentclass[conference]{IEEEtran}
\IEEEoverridecommandlockouts

\usepackage[T1]{fontenc}
\usepackage[utf8]{inputenc}
\usepackage{csquotes}

\usepackage{amsfonts}       
\usepackage{amsmath}
\usepackage{amsthm}%
\usepackage{bbm}
\usepackage{comment}
\usepackage{siunitx}
\usepackage{mathtools}
\usepackage{pifont}
\usepackage{nicefrac}       
\usepackage{xpatch}

\usepackage{aligned-overset}
\usepackage{graphicx, import}
\usepackage{tikz}
\usepackage{pgfplots}
\pgfplotsset{compat=1.18}
\usepackage[compat=1.1.0]{tikz-feynman}
\tikzfeynmanset{warn luatex=false} 
\usetikzlibrary{quantikz2}

\usepackage{tabularray}
\UseTblrLibrary{booktabs}

\usepackage{parskip}
\usepackage{caption}
\usepackage{subcaption} 
 
\usepackage{xcolor}
\usepackage{orcidlink}
\usepackage{soul}
\usepackage{hyperref}
\usepackage{url}            
\hypersetup{colorlinks=true,urlcolor=teal,
            linkcolor=teal,citecolor=teal} 
\usepackage{glossaries-extra}

\usepackage[export]{adjustbox}
\usepackage[inline]{enumitem}
\usepackage[noabbrev,capitalize]{cleveref}
\usepackage{xspace}
\usepackage{censor}

\usepackage{microtype}      

\usepackage{stmaryrd} 
\SetSymbolFont{stmry}{bold}{U}{stmry}{m}{n} 

\usepackage{silence}
\newboolean{anonymous}
\setboolean{anonymous}{false}

\usepackage[
  backend=biber,
  maxcitenames=2,
  minbibnames=3,
  mincitenames=1,
  maxbibnames=3,
  uniquelist=false,
  url=false,
  doi=true,
  eprint=false,
  isbn=false,
]{biblatex}
\setabbreviationstyle[acronym]{long-short}
\glssetcategoryattribute{acronym}{nohyperfirst}{true}

\newacronym{fcc}{FCC}{Fourier coefficient correlation}
\newacronym{fft}{FFT}{fast Fourier transform}
\newacronym{fm}{FM}{feature map}
\newacronym{ftqc}{FTQC}{fault‑tolerant quantum computing}
\newacronym{hea}{HEA}{hardware‑efficient ansatz}
\newacronym{hep}{HEP}{high‑energy physics}
\newacronym{kl}{KL}{Kullback‑Leibler}
\newacronym{lhc}{LHC}{Large Hadron Collider}
\newacronym{lo}{LO}{leading order}
\newacronym{ml}{ML}{machine learning}
\newacronym{mlp}{MLP}{multi‑layer perceptron}
\newacronym{mse}{MSE}{mean squared error}
\newacronym{nisq}{NISQ}{noisy intermediate‑scale quantum}
\newacronym{nlo}{NLO}{next‑to‑leading order}
\newacronym{pca}{PCA}{principal component analysis}
\newacronym{pcc}{PCC}{pearson correlation coefficient}
\newacronym{pqc}{PQC}{parameterised quantum circuit}
\newacronym{qaoa}{QAOA}{quantum approximate optimisation algorithm}
\newacronym{qcd}{QCD}{quantum chromodynamics}
\newacronym{qfm}{QFM}{quantum Fourier model}
\newacronym{qk}{QK}{quantum kernel}
\newacronym{qml}{QML}{quantum machine learning}
\newacronym{qnn}{QNN}{quantum neural network}
\newacronym{rff}{RFF}{random Fourier features}
\newacronym{spam}{SPAM}{state preparation and measurement}
\newacronym{qoc}{QOC}{quantum optimal control}
\newacronym{bp}{BP}{Barren Plateau}
\newacronym{psr}{PSR}{parameter-shift rule}
\newacronym{rwa}{RWA}{rotating-wave approximation}

\newcommand{\melvin}[1]{\todo[color=cyan, inline]{Melvin: #1}}

\newcommand{\wolfgang}[1]{\todo[color=purple, inline]{Wolfgang: #1}}

\makeatletter
\newcommand{\linebreakand}{%
    \end{@IEEEauthorhalign}
    \hfill\mbox{}\par
    \mbox{}\hfill\begin{@IEEEauthorhalign}
}
\makeatother

\newtheorem{theorem}{Theorem}
\newtheorem{proof}{Proof}
\xpretocmd{\endproof}{\raggedright\qed}{}{}

\newtheorem{example}{Example}%
\newtheorem{remark}{Remark}%

\newcommand{\refapx}[1]{\hyperref[#1]{Apx.~\ref*{#1}}}

\newcommand{\eg}{\emph{e.g.}\xspace}
\newcommand{\wrt}{\emph{w.r.t.}\xspace}
\newcommand{\ie}{\emph{i.e.}\xspace}
\newcommand{\cf}{\emph{cf.}\xspace}

\newcommand{\ase}{ansätze\xspace}

\newcommand{\as}{ansatz\xspace}

\newcommand{\btheta}{\ensuremath \boldsymbol{\theta}}
\newcommand{\bomega}{\ensuremath \boldsymbol{\omega}}
\newcommand{\br}{\ensuremath \boldsymbol{r}}
\newcommand{\bp}{\ensuremath \boldsymbol{p}}
\newcommand{\blambda}{\ensuremath \boldsymbol{\lambda}}
\newcommand{\bphi}{\ensuremath \boldsymbol{\phi}}

\newcommand{\bOmega}{\ensuremath \boldsymbol{\Omega}}
\newcommand{\bPi}{\ensuremath \boldsymbol{\Pi}}

\newcommand{\bx}{\ensuremath \boldsymbol{x}}
\newcommand{\bc}{\ensuremath \boldsymbol{c}}

\newcommand{\fullspec}{\ensuremath \tilde{\bOmega}}
\newcommand{\uniquespec}{\ensuremath \bOmega}
\newcommand{\fcc}{\ensuremath \mathtt{FCC}_{\Theta}}

\newcommand{\loss}{\ensuremath \mathcal{L}}

\newcommand{\imag}{\mathrm{i}}
\newcommand{\dt}{\ensuremath \mathrm{d}t}

\newcommand{\X}{\ensuremath{\hat{\mathrm{X}}}}
\newcommand{\Y}{\ensuremath{\hat{\mathrm{Y}}}}
\newcommand{\Z}{\ensuremath{\hat{\mathrm{Z}}}}
\newcommand{\R}{\ensuremath{\hat{\mathrm{R}}}}
\newcommand{\Rot}{\ensuremath{\mathrm{Rot}}}
\newcommand{\RX}{\ensuremath{\hat{\mathrm{R}}_{X}}}
\newcommand{\RY}{\ensuremath{\hat{\mathrm{R}}_{Y}}}
\newcommand{\RZ}{\ensuremath{\hat{\mathrm{R}}_{Z}}}
\newcommand{\CX}{\ensuremath{\mathrm{C}\X}}
\newcommand{\CY}{\ensuremath{\mathrm{C}\Y}}
\newcommand{\CZ}{\ensuremath{\mathrm{C}\Z}}
\newcommand{\CRX}{\ensuremath{\mathrm{C}\RX}}
\newcommand{\CRY}{\ensuremath{\mathrm{C}\RY}}
\newcommand{\CRZ}{\ensuremath{\mathrm{C}\RZ}}
\newcommand{\hadamard}{\ensuremath{\hat{\mathrm{H}}}}
\newcommand{\hamiltonian}{\ensuremath{\hat{H}}}
\newcommand{\tildeH}{\ensuremath{\tilde{H}}}
\newcommand{\dysontime}{\mathcal{T}}
\newcommand{\manifold}{\mathcal{M}}

\definecolor{color1}{HTML}{009682}
\definecolor{color2}{HTML}{DF9B1B}
\definecolor{color3}{HTML}{1f78b4}
\definecolor{color4}{HTML}{002D4C}

\definecolor{lfd1}{HTML}{000000}
\definecolor{lfd2}{HTML}{E69F00}
\definecolor{lfd3}{HTML}{999999}
\definecolor{lfd4}{HTML}{009371}
\definecolor{lfd5}{HTML}{BEAED4}
\definecolor{lfd6}{HTML}{ED665A}
\definecolor{lfd7}{HTML}{1F78B4}

\ifbool{anonymous}{}{}
\ifbool{anonymous}{}{\renewcommand{\blackout}[1]{#1}}
\ifbool{anonymous}{}{}
\ifbool{anonymous}{\newcommand{\genpartemail}[2]{\href{mailto:xxx.xxx@xx.xx}{\blackout{#1}}}}{\newcommand{\genpartemail}[2]{\href{mailto:#1@#2}{#1}}}
\ifbool{anonymous}{\newcommand{\geneurl}[1]{\href{https://thiscatexists.com/}{\blackout{#1}}}}{\newcommand{\geneurl}[1]{\url{#1}}}
\ifbool{anonymous}{\renewcommand{\orcidlink}[1]{}}{}

\newlength{\WIDTH}
\newlength{\HEIGHT}
\def\cornerrad{3pt}

\tikzset{modulearr/.style={-{Triangle[length=0.5em, scale width=0.7]}, draw=lfd1, line width=1pt}}
\tikzset{doublemodulearr/.style={{Triangle[length=0.5em, scale width=0.7]}-{Triangle[length=0.5em, scale width=0.7]}, draw=lfd1, line width=1pt}}
\tikzset{g/.style={rounded corners=\cornerrad}}

\begin{document}

\title{Beyond Gates: Pulse Level Quantum Fourier Models}

\author{
    \IEEEauthorblockN{%
        \blackout{Melvin Strobl}\IEEEauthorrefmark{2}\orcidlink{0000-0003-0229-9897},
        \blackout{Maja Franz}\IEEEauthorrefmark{4}\orcidlink{0000-0002-2801-7192},
        \blackout{Lukas Scheller}\IEEEauthorrefmark{2}\orcidlink{0009-0003-9156-7781},\\
        \blackout{Eileen Kuehn}\IEEEauthorrefmark{2}\orcidlink{0000-0002-8034-8837},
        \blackout{Wolfgang Mauerer}\IEEEauthorrefmark{4}\IEEEauthorrefmark{5}\orcidlink{0000-0002-9765-8313},
        \blackout{Achim Streit}\IEEEauthorrefmark{2}\orcidlink{0000-0002-5065-469X},
    }
 
    \IEEEauthorblockA{
        \IEEEauthorrefmark{2}
        \blackout{Karlsruhe Institute of Technology, Germany},
        \{\genpartemail{melvin.strobl}{kit.edu}, \genpartemail{lukas.scheller}{kit.edu}, \genpartemail{eileen.kuehn}{kit.edu}, \genpartemail{achim.streit}{kit.edu}\}@\blackout{kit.edu}
    }
    \IEEEauthorblockA{
        \IEEEauthorrefmark{4}
        \blackout{Technical University of Applied Science Regensburg, Germany},
        \{\genpartemail{maja.franz}{othr.de}, \genpartemail{wolfgang.mauerer}{othr.de}\}@\blackout{othr.de}
    }
    \IEEEauthorblockA{
        \IEEEauthorrefmark{5}
        \blackout{Siemens AG, Foundational Technology, Munich, Germany}
    }
}

\maketitle

\begin{abstract}
    \unboldmath 
    In the domain of variational quantum algorithms, \glspl{qfm} provide a mathematically well defined structure for \gls{qml}.
    There has been a substantial amount of work on the scalability and trainability of such models showcasing the potential but also the limitations for the prospective application of \glspl{qfm}.
    However, much less is known in the context of pulse-level quantum computing, where the microwave parameters that implement unitary operations on the hardware are used to perform computations directly instead of through the interface of quantum circuits.
    In this work, we evaluate~\glspl{qfm} through the lens of pulse parameters and link metrics such as expressibility and \gls{fcc} to this extended set of variational parameters.
    We show that while control over pulse shapes does not significantly alter the global expressibility or structural correlations of an \as, it fundamentally alters the local optimisation landscape. 
    For composite gates, independent pulse scalings replace a single logical angle by multiple independently tunable sub-angles. 
    This relaxes the rigid monomial couplings induced by the gate-level parameterisation, and provides gradient descent with higher-dimensional escape routes, decoupling local parameter constraints and significantly boosting performance during training.
    Following an analytical proof, we show numerical results validating our theory on training a \gls{qfm} with an exponential (ternary) feature map on a Fourier series with the same frequencies.

\end{abstract}


\glsresetall

\section{Introduction}

In the field of \gls{qml}, \glspl{pqc} are generally used when training models conceptually similar to classical neural networks~\cite{mitarai_quantum_2018, benedetti_parameterized_2019}.
Among these, the class of data-reuploading models~\cite{perezsalinas_data_2020}, also known as \glspl{qfm}, are especially interesting as they impose a mathematically well defined structure on the \gls{pqc}~\cite{schuld_effect_2021}, allowing us to study trainability and dequantisability~\cite{landman_classically_2022,sweke_potential_2025}.
The trainability of PQCs is generally bottlenecked by two distinct phenomena: \glspl{bp}~\cite{ragone_lie_2024} (vanishing gradients globally) and sub-optimal local minima (traps in the loss landscape). 
While much literature focuses on avoiding BPs via structural constraints~\cite{cerezo_cost_2021}, less attention is given to navigating the highly non-convex local landscapes of structured models like QFMs.

In the picture of such \glspl{qfm}, the variance of coefficients, richness of the spectrum and its bias~\cite{duffy_spectral_2025} play an important role for the expressivity of the model~\cite{mhiri_constrained_2024}.

Furthermore, correlations between frequency components coin a fingerprint, unique for each \as~\cite{strobl_fourier_2025} and indicate a limiting factor when training \glspl{qfm}.

Aforementioned works have been focused on the simulation of such models using gate-based circuit structures.
However, running any quantum algorithm on a real device requires
\begin{enumerate*}[label=(\roman*)]
    \item compilation into a series of supported basis gates
    \item application of pulse sequences to the quantum system, describing these gates~\cite{krantz_quantum_2019}.
\end{enumerate*}
As the focus of this work lies on the fundamental optimisation landscape and trainability of QFMs, we restrict our analysis to the coherent pulse dynamics, leaving hardware-specific noise~\cite{greiwe_effects_2023,maschek_make_2025} and calibration overheads for future empirical studies.

In this picture, pulse sequences are subject to additional parameters, independent of the trainable- or input parameters typically seen in \glspl{qfm}.
Fine-tuning of such parameters is subject to \gls{qoc} and carried out by the hardware provider.
Given that one has access to such parameters, they can be treated as an extended set of parameters in the \gls{qml} context.

The literature on quantum-classical algorithms like \gls{qml} and its applications is abundant~\cite{bayerstadler_industry_2021,carbonelli_challenges_2024,yue_challenges_2023}; while
the integration of pulse-level parameters into \gls{qml} architectures has been explored previously, it is much a much less explored aspect.
In Ref.~\cite{liang_napa_2024}, the authors presented an approach of optimising parametrised pulse characteristics for desired operations and \ase, showcasing advantages in expressibility, entangling capability and efficiency with practical benefits demonstrated in various applications.
\citeauthor{magann_pulses_2021}~\cite{magann_pulses_2021} provide a more in-depth overview on how optimal control within \glspl{pqc} necessitates sufficient control resources without which the potential benefits offered by pulse-level control may be limited.
More recently, \citeauthor{acedo_pulsed_2025}~\cite{acedo_pulsed_2025} formulated a pulse-level equivalent of data reuploading models (\glspl{qfm}) showcasing noise robustness on superconducting transmon processors.

\begin{figure}[htb]
    \centering
    \input{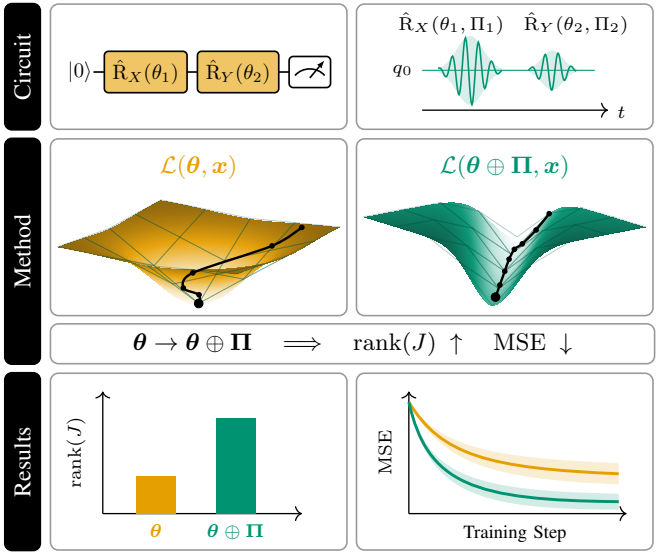}
    \caption{Illustration of our contribution. The optimisation trajectory of a \gls{qfm} in $\mathcal{L}$ can be improved by extending the \textcolor{color2}{gate parameters} $\btheta$ through access to \textcolor{color1}{pulse parameters} $\bPi$.
    This 
    allows to escape local minima and increases the rank of the coefficient Jacobian $\operatorname{rank}(J)$. Practically, the combined optimisation of $\btheta \oplus \bPi$ consistently lowers the \gls{mse} towards a global optimum for all \ase examined in this work.}
    \label{fig:overview}
\end{figure}

In this work, we aim to contribute to a more theoretical understanding of the pulse-level representation of \glspl{qfm} and show that access to pulse parameters can activate or unsuppress coefficients associated with frequencies already present in the encoded frequency set, and can enlarge the algebraic support relative to gate-based \glspl{qfm} and thus significantly improve the training of such models, even beyond over-parametrised \ase obtained through decomposition into basis gates.

Analytical findings are supported through numerical experiments for which we provide the source code for reproduction purposes~\cite{mauerer_1_2022} on Github\footnote{https://github.com/cirKITers/pulse-level-quantum-fourier-models} and in Ref.~\cite{strobl_cirkiters/pulse_2026}.
We illustrate the core contribution of our work in~\autoref{fig:overview}.

\section{Method}
\label{sec:method}

This section describes the theoretical foundations and methodology of this work.
We generally denote vectors by bold symbols; for example, $\btheta$ denotes the vector with components $\theta_k$.

\subsection{Quantum Fourier Models}
\label{subsec:quantum-fourier-models}

Using the same framework as in Ref.~\cite{strobl_fourier_2025}, we define a \gls{qfm} with $L$ layers by
\begin{equation*}
  \hat{U}(\bx, \btheta)
  =
  \hat{W}^{(L+1)}(\btheta)\hat{S}(\bx)\hat{W}^{(L)}(\btheta)\cdots \hat{W}^{(2)}(\btheta)\hat{S}(\bx)\hat{W}^{(1)}(\btheta).
\end{equation*}
Here, $\hat{W}^{(\ell)}$ and $\hat{S}$ are $n$-qubit trainable and encoding unitaries, respectively, which depend on trainable parameters $\btheta$ and input parameters $\bx$.

Using the notation of Ref.~\cite{schuld_effect_2021}, the expectation value of such a model can be written as a truncated Fourier series,
\begin{equation}
    f(\bx, \btheta)
    =
    \sum_{\bomega \in \uniquespec} c_{\bomega}(\btheta)\exp{(\imag \bomega.\bx)},
    \label{eq:fourier-series-simple}
\end{equation}
where $\uniquespec$ denotes the set of unique frequency vectors and $c_{\bomega}(\btheta)$ is the complex-valued Fourier coefficient associated with $\bomega$.
An analytical expression for the Fourier coefficients was given by \citeauthor{wiedmann_fourier_2024}~\cite{wiedmann_fourier_2024}:
\begin{align}
        c_{\omega}(\btheta) = \sum_{\substack{s, c \in \mathbb{N}_{0}^{d}                                                      \\ s^{\prime}, c^{\prime} \in \mathbb{N}_{0}^{w}}}
         \Bigg(&\frac{k_{s, c, s^{\prime}, c^{\prime}}(-\imag)^{\sum_{j=1}^{d} s_{j}}}{2^{\sum_{j=1}^{d}\left(s_{j}+c_{j}\right)}}
        \times p_{s, c}(\omega)\times\nonumber\\
         & \prod_{k=1}^{w} \sin \left(\theta_{k}\right)^{s_{k}^{\prime}} \cos \left(\theta_{k}\right)^{c_{k}^{\prime}}\Bigg).
    \label{eq:exact-coefficients}
\end{align}
Here, $p_{s,c}(\bomega)$ collects all contributions to the coefficient at frequency $\bomega$, while $s,c$ and $s',c'$ count the numbers of sine and cosine factors associated with the input and trainable parameters, respectively.
The constants $k_{s,c,s',c'}$ depend on the chosen basis functions.

Extending the notation of Ref.~\cite{wiedmann_fourier_2024}, it is useful to distinguish three related notions of frequency support: 
(a)~The \emph{encoded frequency set} is the set of frequencies determined by the eigenvalue gaps of the data-encoding generators~--~the set \(\uniquespec\) appearing in \autoref{eq:fourier-series-simple}. Since pulse parameters enter only through the trainable unitaries, they do not change this encoded frequency set. 
(b)~For a fixed family of \ase \(\mathcal A\), not every frequency in \(\uniquespec\) needs to have a non-vanishing coefficient. 
We therefore define the \emph{algebraic coefficient support}
\begin{equation*}
    \uniquespec_{\mathrm{supp}}^{\mathcal A}
    \coloneq
    \left\{
        \bomega \in \uniquespec
        \mid
        c_{\bomega}(\bp) \not\equiv 0
    \right\},
\end{equation*}
where \(\bp\) denotes the trainable parameters of the chosen implementation.
For a gate-level implementation \(\bp=\btheta\), whereas for a pulse-level implementation \(\bp=(\btheta,\bPi)\), or equivalently \(\bp=(\btheta,\blambda)\) when only effective pulse scalings are considered (\cf \autoref{subsec:pulse-parameters-in-quantum-fourier-models}).
Frequencies in \(\uniquespec \setminus \uniquespec_{\mathrm{supp}}^{\mathcal A}\) are allowed by the feature map but vanish for the chosen \as, for example due to algebraic cancellations or symmetries.
(c)~For a parameter-sampling distribution \(\mu\) and threshold
\(\tau>0\), we define the set of \emph{high-variance coefficients} by
\begin{equation}
    \uniquespec_{\mathrm{act}}^{\mathcal A}(\mu,\tau)
    \coloneq
    \left\{
        \bomega \in \uniquespec_{\mathrm{supp}}^{\mathcal A}
        \;\middle|\;
        \operatorname{Var}_{\bp\sim\mu}
        \left[c_{\bomega}(\bp)\right]>\tau
    \right\}.\label{eq:highvar}
\end{equation}
Saying that pulse parameters activate frequencies (\cf \autoref{subsec:coefficient-variance}) means
they can increase the variance of coefficients already contained in
the encoded frequency set \(\uniquespec\), but do not introduce new encoded
frequencies.

We will also use the coefficient-space map
\begin{equation}
    C_{\mathcal A}:\mathcal P \longrightarrow \mathbb C^{|\uniquespec|},
    \qquad
    C_{\mathcal A}(\bp)
    \coloneq
    \bigl(c_{\bomega}(\bp)\bigr)_{\bomega\in\uniquespec},
    \label{eq:coefficient-space-map}
\end{equation}
where \(\mathcal P\) is the parameter domain of the chosen implementation.
Training against a target coefficient vector
\(C^*=(c_{\bomega}^*)_{\bomega\in\uniquespec}\) can be written as
\begin{equation*}
    \loss(\bp)
    =
    \left\|C_{\mathcal A}(\bp)-C^*\right\|_2^2.
\end{equation*}
For real-valued expectation values, the Fourier coefficients obey the Hermitian
symmetry $c_{-\bomega}=\overline{c_{\bomega}}$,
and the same constraint must be imposed on the target coefficients $c^{*}_{\bomega}$.

The spectrum of a \gls{qfm} is determined by eigenvalue differences of the generators appearing in the feature map~\cite{landman_classically_2022}.
For each encoded input component $x_i$, let
\begin{align}
    \Omega_i
    &=
    \left\{
        \Lambda_{i,\boldsymbol{j}} - \Lambda_{i, \boldsymbol{k}}
        \mid
        \boldsymbol{j}, \boldsymbol{k} \in \llbracket 1, d_i \rrbracket^{L}
    \right\},\nonumber\\
    \fullspec &= \Omega_1 \times \cdots \times \Omega_D,
    \label{eq:spectrum}
\end{align}
with multi-indices $\boldsymbol{j}=(j_1,\dots,j_L)$ and $\boldsymbol{k}=(k_1,\dots,k_L)$. 
\begin{equation*}
    \Lambda_{i,\boldsymbol j}
    =
    \sum_{\ell=1}^{L}\lambda_{j_\ell}^{(i,\ell)}.
\end{equation*}
is the sum of eigenvalues selected by the multi-index $\boldsymbol j$ across the $L$ occurrences of the $i$-th encoding generator.

Naturally, frequency-dependent redundancies arise when different eigenvalue-gap combinations generate the same frequency~\cite{landman_classically_2022,jaderberg_let_2023}.
Following \citeauthor{mhiri_constrained_2024}~\cite{mhiri_constrained_2024}, we denote the corresponding \emph{frequency generator} by $R(\bomega)$.
These redundancies induce a feature-map-dependent weighting over $\uniquespec$, and the size of the generator controls the variance of the coefficients via~\cite{mhiri_constrained_2024}
\begin{equation}
    \begin{split}
        \operatorname{Var}_{\btheta}[c_{\bomega}(\btheta)]
        \leq{}
        &\operatorname{Var}_{\text{Haar}}[c_{\bomega}]
        \\
        &+
        \left(
            \frac{C_1\varepsilon}{d^2}
            +
            \frac{C_2\varepsilon}{d(d+1)}
        \right)
        |R(\bomega)|
        \\
        &+
        C_2\frac{\varepsilon^2}{d^2}|R(\bomega)|^2.
    \end{split}
    \label{eq:variance-bound}
\end{equation}
Here, the $\varepsilon$-distance to a $2$-design quantifies how well the ensemble of unitaries generated by uniformly sampling $\btheta$ covers the unitary group.
This directly connects to the notion of \emph{expressibility}~\cite{sim_expressibility_2019}: highly expressible \ase cover the unitary space more uniformly, corresponding to $\varepsilon \to 0$.
Accordingly, for highly expressible models, the coefficient variance approaches the Haar-random limit, whereas restricted expressibility tightens the bound through the multiplicities encoded in $R(\bomega)$.
We refer to~\refapx{app:expressibility-metric} for details on the computation of this metric.

In addition to expressibility, one may quantify global dependencies between Fourier coefficients through the \gls{fcc}.
Since this metric is not needed for the arguments below, we defer its precise definition to~\refapx{app:fourier-coefficient-calculation}.

\subsection{Pulse Level Implementation of Unitary Gates}
\label{subsec:pulse-level-implementation-of-unitary-gates}
Parameterised gates are the basic building blocks of trainable quantum circuits. During transpilation, such gates are decomposed into basis gates supported by the target hardware.
Each basis gate is then realised by a control pulse sequence implementing the corresponding unitary evolution.
It is useful for the below calculations to consider this
in the wider context of a system-level description
of a QPU: For many platforms, notably including superconducting  Transmons, the system Hamiltonian can be modelled as
\(\hamiltonian(t) = \hamiltonian_{\text{static}} + \hamiltonian_{\text{drive}}(t) + \hamiltonian_{\text{dynamic}}(t)\), where 
$\hamiltonian_{\text{static}}$ describes unperturbed energy levels of the qubits and their constant couplings;
$\hamiltonian_{\text{drive}}(t)$ represents external pulses applied to manipulate the qubit states (this is where pulse parameters $\bPi$ enter), and 
$\hamiltonian_{\text{dynamic}}(t)$ accounts for time-varying couplings or environmental interactions, often used to model noise or specific tunable couplers.

We denote a single control pulse by $S_j(t,\bPi)$, where $\bPi$ collects the pulse-shape parameters and $t$ denotes time.
A standard scalar decomposition model (assuming no complex amplitudes or more than one quadrature) is
\begin{equation*}
    S_j(t,\bPi) = E_j(t,\bPi)\cos(\omega_c t + \phi_c),
\end{equation*}
where $E_j(t,\bPi)$ is the pulse envelope, and $\omega_c$ and $\phi_c$ are carrier frequency and carrier phase, respectively.
Combining the control signals with the control Hamiltonians $\hamiltonian_j$ yields the time-dependent Hamiltonian
\begin{equation}
    \hamiltonian(t,\bPi)
    =
    \hamiltonian_{\text{static}}
    +
    \sum_j S_j(t,\bPi)\,\hamiltonian_j,
    \label{eq:hamiltonian}
\end{equation}
which governs the evolution of a state $\ket{\psi(t)}$.

For a time-independent Hamiltonian, the evolution from $t_0$ to $t_0+\delta t$ is (using the convention $\hbar=1$)
$\ket{\psi(t_0+\delta t)} =
e^{-\imag \hamiltonian\delta t}\ket{\psi(t_0)}$.
For the time-dependent Hamiltonian in~\autoref{eq:hamiltonian}, the time evolution of a quantum state is given by ($\dysontime$ denotes time ordering~\cite{efthymiou_qibo_2022,rothesantos_evaluation_2025})
\begin{equation}
    \ket{\psi(t_0 + \delta t)}
    =
    \dysontime
    \exp\left(
        -\imag \int_{t_0}^{t_0 + \delta t}\hamiltonian(t,\bPi)\dt
    \right)
    \ket{\psi(t_0)}.
    \label{eq:time-evolution}
\end{equation}
Typically, this equation must be solved numerically. To balance computational economy and result accuracy, we perform the required computations, but using the common \gls{rwa} as detailed below. We refer to Ref.~\cite{franz_software_2026} for implementation details.
%


\subsection{Pulse Parameters in Quantum Fourier Models}
\label{subsec:pulse-parameters-in-quantum-fourier-models}\label{prop:fourier-params}

Combining~\autoref{subsec:quantum-fourier-models} and~\autoref{subsec:pulse-level-implementation-of-unitary-gates}, we now study the effect of treating pulse parameters $\bPi$ as additional trainable variables during \gls{qfm} optimisation.
Since these parameters appear only in the trainable unitaries, they do not alter frequencies in the exponent of~\autoref{eq:fourier-series-simple}; they only modify Fourier coefficients.
This allows us to compare gate-level and pulse-level optimisation while keeping the represented frequency set fixed.

Among our numerical results, it is useful to consider a simplified fixed-axis pulse-area model to guide intuition on the observed phenomena, and to motivate our chosen diagnostics. 
We employ the canonical \gls{rwa}~\cite{krantz_quantum_2019} to absorb a rapidly oscillating carrier into an effective drive Hamiltonian. 
Gate actions are then governed primarily by  pulse envelope and phase, while the generator $\tildeH_j$ is fixed by the rotation axis. 
Pulse parameters then act as effective angle scalings.

To isolate the main mechanism, consider first a single-axis rotation in the effective rotating frame.
Assume that the corresponding effective Hamiltonian takes the form
\begin{equation*}
    \hamiltonian_{\text{eff}}(t, \theta, \bPi)
    =
    \theta E(t,\bPi) \tildeH_{j},
\end{equation*}
where $\tildeH_j$ is time independent.\footnote{The scalar pulse-area model is appropriate when pulse parameters modify only the amplitude or duration of a resonant fixed-axis control (\ie, the pulse changes only the total rotation angle). 
In more general pulse models that consider, for instance, detuning, time-dependent phases, leakage or crosstalk, pulse parameters can affect several independent physical contributions to the implemented operation. 
In that case, $\lambda(\bPi)$ would need to be replaced by a vector, which we do not consider in the present paper. 
Pulse optimisation could then change not only the size of the rotation, but also the effective axis or phase, or introduce multi-qubit interaction terms generated by the pulse.}
Then the resulting unitary (assuming fixed-axis resonant control and commuting effective Hamiltonians $[\hamiltonian_{\text{eff}}(t), \hamiltonian_{\text{eff}}(t')]=0$)\footnote{Laboratory-frame Hamiltonians do not to commute at different times.
In the rotating frame with the \gls{rwa} applied, a resonant fixed-phase single-axis drive is described by an effective Hamiltonian $\hamiltonian_{\mathrm{eff}}(t)=\Omega(t)G$, where $G$ is time independent.
Only in this effective description is the time ordering trivial, and the pulse affects the gate primarily through the pulse area $\int \Omega(t)\dt$.} is
\begin{equation}
    \begin{aligned}
      \hat{U}(\theta,\bPi)
      &=
      \exp\left(
          -\imag\theta \int_{0}^{\delta t} E(\tau, \bPi)\tildeH_{j}\mathrm{d}\tau
      \right)
      \\
      &=
      \exp\left(
          -\imag\theta\lambda(\bPi)\tildeH_{j}
      \right),
    \end{aligned}
    \label{eq:unitary-rectangular-pulse}
\end{equation}
 with the effective pulse-area factor
\begin{equation*}
    \lambda(\bPi) \coloneq \int_{0}^{\delta t} E(\tau,\bPi)\mathrm{d}\tau .
\end{equation*}
For a rectangular envelope $E(t, A)=A$ on $[0, \delta t]$, this reduces to $\lambda(\bPi)=A\,\delta t$.
Any constant prefactors arising from calibration or the \gls{rwa} can be absorbed into $\lambda(\bPi)$.

Thus, for a single basis rotation, the pulse parameter acts only through the effective angle $\theta\lambda(\bPi)$.
In this simplified setting, optimising $\bPi$ amounts to a reparameterisation of the original gate parameter.
This remains valid if $\delta t$ itself is included in $\bPi$: its effect is still mediated through the scalar factor $\lambda(\bPi)$.

For more general envelopes, the dependence on $\bPi$ need not reduce exactly to a single scalar factor at the laboratory-frame level.
However, in the effective single-axis description the dominant effect is still a pulse-dependent rescaling of the rotation angle, and the argument below carries through with only notational modifications.
For example, using a textbook calculation~\cite{arfken_mathematical_1985} for a Gaussian envelope $E(t,A,\sigma)=A\,\exp\!\left(-\nicefrac12\left(\frac{t-t_c}{\sigma}\right)^2\right)$, we find 
\begin{align}
\lambda(A,\sigma) &=
\int_{0}^{\delta t} E(t, A,\sigma)\text{d}t\nonumber \\&
= 
A\sigma\sqrt{\frac{\pi}{2}}
\left[
\operatorname{erf}\left(
\frac{\delta t-t_c}{\sqrt{2}\sigma}
\right)
+
\operatorname{erf}\left(
\frac{t_c}{\sqrt{2}\sigma}
\right)
\right]\\
&\approx A\sigma\sqrt{2\pi}
\end{align}
where the approximation holds for a centered pulse ($t_{c}=\nicefrac{\delta t}{2}$) that is contained in the pulse window ($\delta t\gg \sigma$).

The situation changes for composite gates.
There, a single logical gate parameter $\theta$ controls several sub-rotations, while the pulse parameters of the sub-gates can vary independently.
This distinction is illustrated by the following example.

\begin{example}[Composite gate with independent pulse parameters]
    \label{ex:crx-example}
    Assume the basis gate set $\left\{\RX,\RY,\RZ,\CX\right\}$ and the decomposition
    \begin{equation*}
        \CRX(\theta)
        =
        \RZ(\nicefrac\pi2)\cdot \RY(\nicefrac\theta2)\cdot \CX \cdot \RY(-\nicefrac\theta2)\cdot \CX \cdot \RZ(-\nicefrac\pi2).
    \end{equation*}
    Let the $m$-th basis gate (out of $M$ sub-gates) carry its own pulse parameter $\bPi_m$, and let $\lambda_m(\bPi_m)$ denote the corresponding effective scaling.
    Then the implemented unitary may be written as
    \begin{equation*}
        \hat{U}_{\mathrm{CRX}}(\theta,\{\bPi_m\}_{m=1}^M)
        =
        \prod_{m=1}^{M}
        U_m\!\bigl(\alpha_m(\theta),\bPi_m\bigr),
    \end{equation*}
    where $\alpha_m(\theta)$ is the sub-gate angle induced by the global logical parameter $\theta$.
    Using $e^{-\imag \frac{\gamma}{2}P}=\cos(\frac{\gamma}{2})\mathbbm{1}-\imag\sin(\frac{\gamma}{2})P$, each sub-gate contributes factors of the form
    \begin{align*}
    &\sin\left(\alpha_m(\theta)
      \lambda_m(\bPi_m)/2\right),
    &\cos\left(\alpha_m(\theta)
     \lambda_m(\bPi_m)/2\right).
    \end{align*}
    Since a single $\theta$ controls several distinct angles $\alpha_m(\theta)$, while the $\lambda_m(\bPi_m)$ are independent, it is generically not reducible to a single reparametrised angle.
    In particular, on the control branch $\ket{0}$, the standard decomposition yields
    \begin{equation*}
        \RZ(\nicefrac\pi2) \RY(\nicefrac\theta2) \RY(-\nicefrac\theta2) \RZ(-\nicefrac\pi2)=\mathbbm{1}.
    \end{equation*}
    If the two $\RY$ sub-gates are independently rescaled by $\lambda_1$ and $\lambda_2$, then the same branch becomes
    \begin{equation*}
        \begin{aligned}
            &\RZ(\nicefrac\pi2)
            \RY\left(\nicefrac{\theta}{2}\lambda_1\right)
            \RY\left(-\nicefrac{\theta}{2}\lambda_2\right)
            \RZ(-\nicefrac\pi2)
            \\
            &\qquad =
            \RX\left(-\nicefrac{\theta}{2}(\lambda_1-\lambda_2)\right).
        \end{aligned}
    \end{equation*}
    Hence, if $\lambda_{1} \neq \lambda_{2}$, the operation on the $\ket{0}$ branch is no longer the identity.
    The resulting controlled operation therefore leaves the original $\CRX$ gate family, even though it remains block diagonal in the computational basis.
    This demonstrates that pulse-level flexibility in a composite implementation spans a larger local unitary family than simple angle reparameterisation of the logical gate.
\end{example}

More generally, if $U(\theta,\bPi)$ is a pulse-level realisation of a single-axis basis rotation, then it is unitarily equivalent to $\R_j(\theta\lambda(\bPi))$.
By contrast, for a composite logical gate
\begin{equation*}
    \hat{U}_{\text{comp}}\left(\theta, \{\bPi_m\}_{m=1}^{M}\right)
    =
    \prod_{m=1}^{M} U_m \bigl(g_m(\theta), \bPi_m\bigr),
\end{equation*}
the Pauli decomposition contains terms of the form
\begin{equation*}
    \prod_{m=1}^{M}
    \sin\bigl(g_m(\theta)\lambda_m(\bPi_m)\bigr)^{a_m}
    \cos\bigl(g_m(\theta)\lambda_m(\bPi_m)\bigr)^{b_m},
\end{equation*}
with some sub-gate specific factors $g_m$, $a_m$ and $b_m$.
Whenever $M\ge 2$, at least two sub-gates act on different axes, and the $\lambda_m(\bPi_m)$ are independently adjustable, these factors are not reducible to a function of a single effective angle.

To clarify this mechanism further, let a \gls{qfm} contain $w$ trainable basis rotations, and let the $k$-th of these gates be implemented in the effective single-axis form
\begin{equation*}
    \hat U_k(\theta_k,\bPi_k)
    =
    \exp \bigl(-\imag \theta_k\lambda_k(\bPi_k) \tildeH_k\bigr),
\end{equation*}
where $\tildeH_k$ is the corresponding generator.
Consider the effective angle $\varphi_k := \theta_k\lambda_k(\bPi_k)$, then the Fourier coefficients retain the monomial structure of Ref.~\cite{wiedmann_fourier_2024} and may be written as
\begin{equation}
    \begin{aligned}
        c_{\bomega}(\btheta,\bPi)
        =
        \sum_{(s,c,s',c')}\Big(
        &\kappa_{s,c,s',c'}\,p_{s,c}(\bomega)
        \times\\
        &
        \prod_{k=1}^{w}
        \sin(\varphi_k)^{s'_k}
        \cos(\varphi_k)^{c'_k}\Big),
    \end{aligned}
    \label{eq:fourier-series-expanded}
\end{equation}
where $\kappa_{s,c,s',c'}$ absorb constants and phases from~\autoref{eq:exact-coefficients}.

In particular, the map
\begin{equation*}
    (\btheta,\bPi)\mapsto \bphi=(\varphi_1,\dots,\varphi_w)
\end{equation*}
has image in $\mathbb{R}^w$, so basis-gate pulse parameters alone do not enlarge the reachable coefficient set; they only change its parametrisation.
Equivalently, the Jacobian
\begin{equation*}
    \frac{\partial \bphi}{\partial(\btheta,\bPi)}
\end{equation*}
has rank at most $w$ in a $(w+|\bPi|)$-dimensional parameter space and therefore possesses a non-trivial kernel.

Now let a logical trainable gate $k$ be implemented as a composite gate
\begin{equation*}
    \hat U_{\mathrm{comp},k}\bigl(\theta_k,\{\bPi_{k,m}\}_{m=1}^{M_k}\bigr)
    =
    \prod_{m=1}^{M_k}
    \hat{U}_{k,m} \bigl(g_{k,m}(\theta_k), \bPi_{k,m}\bigr).
\end{equation*}
Then its Pauli decomposition contains factors of the form
\begin{align}
    \prod_{m=1}^{M_k}
    &\sin\bigl(g_{k,m}(\theta_k)\lambda_{k,m}(\bPi_{k,m})\bigr)^{a_{k,m}}\times\nonumber\\
    &\cos\bigl(g_{k,m}(\theta_k)\lambda_{k,m}(\bPi_{k,m})\bigr)^{b_{k,m}}\label{eq:decomp_multi}.
\end{align}
These terms are usually not reducible to a function of a single effective logical angle.
Therefore, at the logical-gate level, independent pulse parameters can relax the rigid monomial coupling present in the gate-based representation.

\subsection{Effects on Learning}
In standard gate-based QML, escaping poor stationary points often requires over-parameterisation, for instance by increasing circuit depth~\cite{larocca_theory_2023}.
The above consideration shows that pulse-level control provides a hardware-native analogue of this approach:
for basis gates it changes the optimisation parameterisation, whereas for composite gates it can introduce additional local directions in coefficient space without changing the encoded frequency set.

As~\autoref{eq:decomp_multi} shows, pulse parameters can partially decouple terms that are rigidly linked at the logical-gate level.
For non-rectangular envelopes, additional pulse-shape dependence generally produces further correction terms; this does not alter the encoded spectrum, but it can increase the local flexibility of the coefficient map.

\begin{remark}
    \label{rem:expressibility-local-vs-global}
    Pulse parameters need not substantially change global metrics such as expressibility, even when they improve trainability.
    For the family of \ase considered here, their primary effect is to modify the parameterisation of the trainable unitary family and, for composite gates, to enlarge local tangent directions accessible during optimisation.
    Accordingly, pulse-level control can smooth the local loss landscape without necessarily changing the global set of representable frequencies.
\end{remark}

\subsection{Training and Trainability}
\label{subsec:training-and-trainability}

We consider a target function $f^*$ whose frequency support matches that of the \gls{qfm}, with coefficients $C^*$ (\cf \autoref{subsec:quantum-fourier-models}) drawn randomly from the unit disk.
In this subsection, we restrict to the one-dimensional input case for simplicity.
Assuming that the frequencies are orthogonal on $[0,2\pi]$ (for example, $\Omega \subset \mathbb{Z}$), Parseval's identity gives
\begin{equation}
    \begin{aligned}
        \loss(\btheta)
        &=
        \frac{1}{2\pi}\int_{0}^{2\pi}\bigl|f(x,\btheta)-f^{*}(x)\bigr|^2 \mathrm{d}x
        \\
        &=
        \sum_{\omega\in\Omega}
        \bigl|c_\omega(\btheta)-c^{*}_{\omega}\bigr|^2.
    \end{aligned}
    \label{eq:loss}
\end{equation}
Thus, training can be viewed as matching the model coefficients to the target coefficients in Fourier space.

Following~\autoref{eq:exact-coefficients} and Ref.~\cite{mhiri_constrained_2024}, different frequencies generally can depend on the same gate parameter \(\theta_k\), which creates coupled constraints across coefficient directions.
Using the coefficient-space map from~\autoref{eq:coefficient-space-map}, the gate-level Jacobian is
\begin{equation}
    J_{\theta}
    \coloneq
    \frac{\partial C_{\mathcal A}}{\partial \btheta}
    \in \mathbb C^{|\uniquespec|\times |\btheta|}.
    \label{eq:gate-level-jacobian}
\end{equation}
Since the trainable parameters are real while the Fourier coefficients are complex, all ranks and tangent-space dimensions below are understood over \(\mathbb R\) with \(\mathbb C^{|\uniquespec|}\cong \mathbb R^{2|\uniquespec|}\).

If we define the residual vector
\begin{equation*}
    \br
    \coloneq
    C_{\mathcal A}(\btheta)-C^*,
    \qquad
    r_{\bomega}
    =
    c_{\bomega}(\btheta)-c^*_{\bomega},
\end{equation*}
then a first-order critical point of \(\loss\) satisfies
\begin{equation*}
    \operatorname{Re}
    \left(
        \br^\dagger
        \frac{\partial C_{\mathcal A}}{\partial \theta_k}
    \right)
    =
    \sum_{\bomega\in\uniquespec}
    \operatorname{Re}
    \left(
        r_{\bomega}^{*}
        \frac{\partial c_{\bomega}}{\partial \theta_k}
    \right)
    =
    0
    \qquad
    \forall k.
\end{equation*}
This condition is automatically satisfied at a global minimum, where \(C_{\mathcal A}(\btheta)=C^*\). 
More generally, it is also satisfied at non-global stationary points whenever the residual is orthogonal to all tangent directions generated by the columns of \(J_{\theta}\).

To build geometric intuition, assume that \(|\uniquespec|>|\btheta|\). 
Then the coefficient-space map
\begin{equation*}
    C_{\mathcal A}:\mathbb R^{|\btheta|}
    \longrightarrow
    \mathbb C^{|\uniquespec|}
    \cong
    \mathbb R^{2|\uniquespec|}
\end{equation*}
has an image that, locally and at regular points, is an immersed manifold \(\manifold\) of dimension
\begin{equation*}
    \dim T_{\bc}\manifold
    =
    \operatorname{rank}_{\mathbb R} J_{\btheta}
    \leq
    |\btheta|.
\end{equation*}
Local optimisation is equivalent to moving on this coefficient manifold to reduce the Euclidean distance to the target point \(C^*\).
Because the coefficients are trigonometric polynomials in the parameters,
\begin{equation*}
    c_\omega(\btheta)
    =
    \sum_{\ell}
    \kappa_\ell
    \prod_{k=1}^{w}
    \sin(\theta_k)^{s'_{\ell,k}}
    \cos(\theta_k)^{c'_{\ell,k}},
\end{equation*}
the manifold \(\manifold\) is generally highly curved and may exhibit self-intersections, singular points, or narrow tangent directions.

As in~\autoref{subsec:pulse-parameters-in-quantum-fourier-models}, two distinct cases arise when pulse parameters are introduced and we optimise
\begin{equation*}
    \loss(\btheta, \bPi)
    \coloneq
    \|\bc(\btheta, \bPi)-\bc^*\|^2.
\end{equation*}

\paragraph{Basis gates}
If each trainable gate remains an effective single-axis basis rotation, then the reachable coefficient manifold does not change; only its parameterisation changes.
By the chain rule,
\begin{equation}
    \frac{\partial \loss}{\partial \Pi_k}
    =
    \frac{\partial \loss}{\partial \varphi_k}\,
    \theta_k
    \frac{\partial \lambda_k}{\partial \Pi_k},
    \label{eq:gradient-expanded}
\end{equation}
where $\varphi_k=\theta_k\lambda_k(\Pi_k)$.
Hence, pulse parameters act as parameter-dependent rescaling or preconditioning of the gradient field (given that
$\theta_{k}\neq 0$ and $\partial \lambda_{k}/\partial \Pi_{k}\neq 0$ as otherwise, the pulse parameters provide no first-order direction).
They can improve conditioning and alter optimisation trajectories, but they do not enlarge the set of reachable coefficients.

\paragraph{Composite gates}
If a logical gate is implemented by several basis gates with independently adjustable pulse scalings, then additional tangent directions can appear in coefficient space. 
Let \(\blambda\) denote the collection of effective pulse scalings with $J_{\lambda}(\blambda) \coloneq \frac{\partial C_{\mathcal A}}{\partial \blambda}$ and define the pulse-level Jacobian
\begin{equation}
    J_{\mathrm{ext}}(\btheta,\blambda)
    \coloneq
    \frac{\partial C_{\mathcal A}}{\partial(\btheta,\blambda)}
    =
    \begin{bmatrix}
        J_{\theta} & J_\lambda
    \end{bmatrix}.
    \label{eq:extended-jacobian}
\end{equation}
The tangent spaces accessible to first-order optimisation are\footnote{The gate-level tangent space is the space of all first-order
coefficient variations induced by infinitesimal changes of the
gate parameters. More precisely, at a point \(\theta_0\), let
\(\theta(t)=\theta_0+t v\) with \(v\in\mathbb{R}^{|\theta|}\).
Then
\begin{equation*}
C_{\mathcal{A}}(\theta(t)) = C_A(\theta_{0}) + t J_\theta(\theta_{0})v + \mathcal{O}(t^2).
\end{equation*}
Thus the accessible first-order directions in coefficient space are
\begin{equation*}
T_{\text{gate}}(\theta_{0}) \coloneq
\{J_\theta(\theta_{0})v \mid v\in\mathbb{R}^{|\theta|}\}
= \operatorname{Im} J_{\theta}(\theta_{0}).
\end{equation*}
Since the coefficients are complex-valued and the trainable
parameters are real, this tangent space is understood as a real
linear subspace of \(\mathbb{C}^{|\Omega|}\simeq\mathbb{R}^{2|\Omega|}\).}
\begin{align*}
    T_{\text{gate}}  &=
    \operatorname{Im}(J_{\theta}),
    &T_{\text{pulse}} &=
    \operatorname{Im}(J_{\text{ext}}).
\end{align*}
Consequently, the increase in local tangent-space dimension is
\begin{equation}
    \dim T_{\text{pulse}}
    -
    \dim T_{\text{gate}}
    =
    \operatorname{rank}_{\mathbb R}J_{\text{ext}}
    -
    \operatorname{rank}_{\mathbb R}J_{\theta} .
    \label{eq:manifold-increase}
\end{equation}
This implies that pulse parameters provide genuinely new local search directions when
\begin{equation*}
    \operatorname{rank}_{\mathbb R}J_{\text{ext}}
    >
    \operatorname{rank}_{\mathbb R}J_{\theta}.
\end{equation*}
Equivalently, at least one pulse derivative
\(\partial C_{\mathcal A}/\partial\lambda_{k,m}\) must have a component outside
$\operatorname{Im}(J_{\theta})$.

If \(q\) independent pulse-scaling variables are introduced, the general upper
bound is
\begin{equation*}
    0
    \leq
    \operatorname{rank}_{\mathbb R}J_{\mathrm{ext}}
    -
    \operatorname{rank}_{\mathbb R}J_{\theta}
    \leq
    \min\left\{
        q,\,
        2|\uniquespec|
        -
        \operatorname{rank}_{\mathbb R}J_{\theta}
    \right\}.
\end{equation*}
For a composite gate with \(M_g\) independently scaled sub-gates, a common rescaling of all sub-rotations is often equivalent to changing the original logical gate angle. 
In such cases, at most \(M_g-1\) directions per composite gate can be genuinely new. 
However, the actual increase is not determined solely by the number of sub-gates or rotation axes; it is always given by the
rank difference in~\autoref{eq:manifold-increase}.

Composite gates can therefore enlarge the local tangent space seen by gradient descent even though the encoded frequency set \(\uniquespec\) is unchanged.
In practice, most \ase contain both types of gates.
Basis-gate pulse parameters alter the optimisation metric on the same manifold, while composite-gate pulse parameters can add genuinely new local search directions.

Because~\autoref{eq:loss} is the squared distance from $\bc^*$ to the coefficient manifold, stationary points occur when the residual vector is orthogonal to the tangent space.
Enlarging the tangent space therefore makes it less likely that a non-zero residual remains orthogonal to all available descent directions.

\begin{example}
    Consider a logical gate with a single parameter $\theta$ and decomposition
    \begin{equation*}
        \hat{U}_{\mathrm{std}}(\theta)=\RX(a\theta) \RY(b\theta).
    \end{equation*}
    In the standard parameterisation, the pair of effective sub-angles $(a\theta, b\theta)$ lies on a one-dimensional line in $\mathbb{R}^2$, as the ratio of angles remains invariant with scaling \(\theta\).
    If the two sub-gates acquire independent pulse scalings, the implemented gate becomes
    \begin{equation*}
        \hat{U}_{\mathrm{pulse}}(\theta,\lambda_1,\lambda_2)
        =
        \RX(a\theta\lambda_1) \RY(b\theta\lambda_2),
    \end{equation*}
    so the pair of effective sub-angles becomes $(a\theta\lambda_1,b\theta\lambda_2)$.
    Near generic points with \(\theta\neq 0\), this spans a two-dimensional region rather than a one-dimensional line,
    as \(\lambda_{i}\) can be independently scaled.
    Thus, pulse parameters create local directions that are unavailable in the logical gate parameter alone. In terms of the above considerations, the change is from rank 1 to rank 2 in the Jacobian.
\end{example}

This observation leads to the following statement.

\begin{theorem}[Escape directions from non-global stationary points]
    \label{theo:local-minima}
    Consider a \gls{qfm} trained on a Fourier-compatible target
    \begin{equation*}
        f^*(x)=\sum_{\omega} c^{*}_{\omega} e^{\imag \omega x},
    \end{equation*}
    and let $\btheta^*$ be a non-global critical point of the gate-level loss $\loss(\btheta)$ with residual $\br\neq 0$.
    Assume that, at $(\btheta^{*}, \blambda=\mathbf{1})$, the derivatives with respect to at least one composite-gate pulse scaling produce a coefficient-space direction that is not contained in the span of the gate-level Jacobian columns.
    Then, for Lebesgue-almost all target coefficient vectors $\bc^*$, the point $(\btheta^*,\blambda=\mathbf 1)$ is not a critical point of the extended loss $\loss(\btheta,\blambda)$.
    In particular,
    \begin{equation*}
        \nabla_{\blambda}\loss\big|_{\blambda=\mathbf 1}\neq 0,
    \end{equation*}
    so the extended parameterisation provides a first-order escape direction.
\end{theorem}

\autoref{eq:manifold-increase} gives the geometric mechanism behind this result: composite-gate pulse parameters can enlarge the local tangent space.
\autoref{theo:local-minima} then states that, under a generic non-degeneracy assumption, these additional directions almost surely destroy non-global criticality for randomly drawn targets.
A proof is provided in the following.

\begin{proof}[Proof of \autoref{theo:local-minima}]
    \label{proof:local-minima}
    At $(\btheta^*, \blambda = \mathbf{1})$, the gradient along a composite-gate pulse scaling $\lambda_{k,m}$ is:
    \begin{equation}
        \frac{\partial\mathcal{L}}{\partial\lambda_{k,m}}\bigg|_{\blambda=\mathbf{1}} = 2\,\text{Re}\!\left[\sum_\omega r^*_\omega \frac{\partial c_\omega}{\partial\lambda_{k,m}}\bigg|_{\blambda=\mathbf{1}}\right]
    \end{equation}
    For a composite gate with sub-gates on distinct axes, $\frac{\partial c_\omega}{\partial\lambda_{k,m}}$ is not in $\text{range}(J_{\theta})$ (since it generates cross-axis Pauli components absent in the $\btheta$-gradient).
    The condition $\nabla_{\blambda}\mathcal{L} = 0$ then requires $\mathbf{r} \perp V$ where $V = \text{range}(J_{\theta}) + \text{span}\big\{\frac{\partial\mathbf{c}}{\partial\lambda_{k,m}}\big\}$ strictly contains $\text{range}(J_{\theta})$.

    Since $\btheta^*$ is a local minimum, we have $\mathbf{r} \perp \text{range}(J_{\theta})$.
    The additional condition $\mathbf{r} \perp \text{span}\big\{\frac{\partial\mathbf{c}}{\partial\lambda_{k,m}}\big\}$ constrains $\mathbf{r}$ to a strict subspace of $\text{range}(J_{\theta})^\perp$.
    For $\mathbf{c}^*$ drawn from a continuous distribution, $\mathbf{r} = \mathbf{c}(\btheta^*) - \mathbf{c}^*$ lies in this subspace with Lebesgue measure zero.
\end{proof}

While this provides a formal motivation for pulse-level optimisation, it also increases the number of trainable parameters and hence the optimisation complexity.
The exact pulse-parameter counts used in our experiments are listed in~\autoref{tab:num-of-pp-per-gate}.
Note that, while pulse parameters smooth the landscape by eliminating specific constrained traps inherent to gate-level compilation, they do not guarantee a trap-free landscape unless the system becomes fully over-parametrised.

\subsection{Effect on the FCC and Expressibility}
\label{subsec:effects-on-fcc-and-expressibility}

Metrics such as the \gls{fcc} and expressibility are global properties: they are computed by sampling parameters over their full domain to characterise the overall distribution of accessible unitaries or coefficients.

For the family of \ase and pulse-parameter ranges considered here, the dominant observed effect is local rather than global. 
We therefore do not expect, and numerically do not observe, a large change in FCC or expressibility.

The same reasoning indicates why improved trainability need not be accompanied by a larger \gls{fcc}.

The \gls{fcc} measures normalised linear correlations between coefficients, whereas pulse parameters can increase the variance of individual coefficients without substantially changing the underlying correlation structure.
In particular, since the Pearson correlation normalises by the coefficient variances, unsuppressing previously weak coefficients need not increase the global \gls{fcc}.

A key point is that pulse parameters do not introduce new frequencies.
Rather, they can unsuppress existing frequencies by breaking rigid algebraic cancellations inherited from composite gate decompositions; see~\autoref{prop:fourier-params}.
This improves local trainability while leaving the global frequency support unchanged.

Likewise, the variance bound of~\citeauthor{mhiri_constrained_2024}~\cite{mhiri_constrained_2024} relates coefficient variance to expressibility through an upper bound.
A circuit need not saturate that bound.
Pulse parameters can therefore move coefficient variances closer to their admissible upper range without implying any substantial change in the global expressibility of the \as family.

\section{Results}
\label{sec:results}

In this section we provide numerical results to support our findings from~\autoref{sec:method}.
Notably, the pulse shape $E$ is also considered a hyperparameter in \gls{qoc} and while usually a DRAG~\cite{motzoi_simple_2009} pulse is used in practice, we consider a Gaussian pulse described as $A\,\exp\!\left(-\nicefrac12\left(\frac{t-t_c}{\sigma}\right)^2\right)$, where $t_c$ is the central pulse time and $\bPi \coloneqq \{A, \sigma\}$, with $A$ and $\sigma$ being the pulse amplitude and pulse width, respectively.
While we do not constrain the range of values for the pulse parameters in the scope of this work, we acknowledge that for a physical implementation, the amplitude would require bounded optimisation $A < A_{\max}$ to prevent leakage.
Furthermore we choose a set of basis gates as $\RX$, $\RY$, $\RZ$ and $\CZ$.
We refer to Ref.~\cite{franz_software_2026} for more details regarding the pulse level simulation implementation.
All numerical simulations in this section represent closed-system dynamics without noise or decoherence, to isolate the theoretical impact of the optimisation landscape.
Finally, we apply the \gls{rwa} for all numerical experiments presented in this section; supplementary simulations utilizing the exact interaction-picture dynamics produced nearly identical results.

In the following experiments, we consider a model as defined in~\autoref{subsec:quantum-fourier-models} with $3$ qubits and a two trainable unitaries.
For the encoding we choose a ternary feature map~\cite{peters_generalization_2022} where $x$ is embedded through Pauli rotations defined as:
\begin{equation*}
    \hat{S}(x) = \bigotimes_{m=1}^n S(x 3^m).
\end{equation*}
where $n$ depicts the number of qubits.
This feature map is chosen as it implements an exponential feature map with a gap free spectrum.

\subsection{Training}
\label{subsec:training}

In this section we present the numerical results for the training of \glspl{qfm}.

We train a \gls{qfm} using the Adam optimiser (with analytical gradients) on the \gls{mse} loss against a Fourier series dataset. This dataset is generated using the exact same frequencies that the model theoretically represents, with coefficients sampled randomly within the unit disk.
Discrete datapoints are sampled according to the Nyquist sampling theorem~\cite{shannon_communication_1949}.
The results of this experiment are shown in~\autoref{fig:mse-over-circuits} for all of the \ase considered in this work (see Ref.~\cite{strobl_qml_2025} for implementation details), and sorted by their number of pulse parameters.
We show the \gls{mse} for both the purely gate-based training, the case where we enable access to the pulse parameters and a third scenario in which we decompose all gates into basis gates and add additional trainable parameters to each decomposed gate.
This decomposed circuit thus has almost the same number of trainable parameters as in case of the pure pulse simulation while still being fully gate based.
The standard deviation over ten seeds for a combined model and data initialisation is reported via error bars.
In addition to the \gls{mse}, we also report the distance between the rank of the coefficient Jacobian $\operatorname{rank}(J_{\theta}) - \operatorname{rank}(J_\text{ext})$ (\cf~\autoref{eq:extended-jacobian}) between the gate-based and pulse-based simulation.

\begin{figure}[htb]
    \includegraphics[width=\columnwidth]{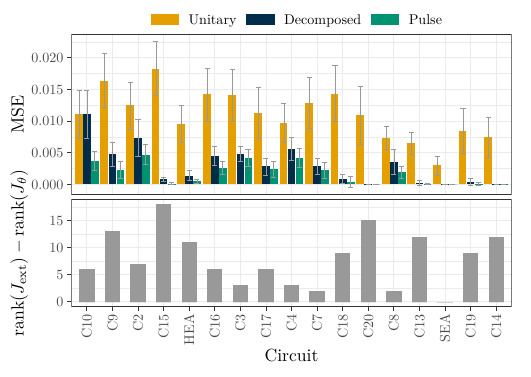}
    \caption{\gls{mse} for all \ase examined in this work, trained either only using unitary gate parameters (without and with a decomposition into basis gates), or in combination with pulse parameters. Results represent the average over ten different seeds used for both the model and data initialisation. Standard deviation is reported via error bars.}
    \label{fig:mse-over-circuits}
\end{figure}


The results confirm our theory from~\autoref{sec:method} that access to the pulse parameters results in a decrease of the \gls{mse} for all \ase.
Moreover, the results indicate a correlation between the number of pulse parameters and decrease of the \gls{mse}, thus confirm our expectation that composed gates contribute more significantly than basis gates.
The reported \gls{mse} of the decomposed circuit falls between the two other cases, but is for almost all circuits close to the pulse simulation, suggesting that over-parameterisation is a significant contributor, while still not being on-par with pulse simulation.
%
%
%
For the sake of completeness, we provide the full training history in~\refapx{app:additional-results}.

\subsection{Coefficient Variance}
\label{subsec:coefficient-variance}

In this section we present the numerical results for the activation of frequency components due to access to pulse parameters.
For this purpose, we instantiate a \gls{qfm} identical to the one we used in \autoref{subsec:training}.
To obtain the coefficients of such a circuit we sweep the input $x$ in a range from $0$ to $2\pi$ and then repeat this process for $4000$ different sets of gate parameters.
By applying a normal distribution on the scaler of pulse parameters with unit mean, we can observe the effect on coefficients when gradually increasing the variance $\sigma_{\blambda}$.
In~\autoref{fig:num-frequencies-over-distortion} we present the results of this experiment, by reporting the number of active frequencies.

\begin{figure}[htb]
    \includegraphics[width=\columnwidth]{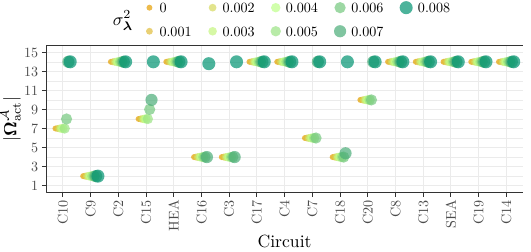}
    \caption{Count of active frequencies, calculated by the number of high-variance coefficients (see~\autoref{eq:highvar}) exceeding a threshold $\tau > 5 \cdot 10^{-6}$. Pulse parameter variance is indicated by color.}
    \label{fig:num-frequencies-over-distortion}
\end{figure}

As the variance of a particular iteration is depicted by color, we can note that different \ase (sorted by number of pulse parameters) are affected very differently by this distortion.
Note that non-integer number of frequencies appear due to the average over $10$ different seeds used for the initialisation.

While~\autoref{fig:num-frequencies-over-distortion} demonstrates that access to pulse parameters can increase the variance of coefficients and activates previously suppressed frequencies, \autoref{fig:expr-fcc-over-distortion} shows that the \gls{fcc} remains largely invariant. 
This apparent discrepancy is theoretically consistent as mentioned in~\autoref{subsec:effects-on-fcc-and-expressibility}: pulse parameters break exact algebraic cancellations within composite gates (\eg, imperfect cancellation of sub-rotations, \cf~\autoref{ex:crx-example}), thereby increasing the variance of specific Fourier coefficients. 
However, the global linear correlations captured by the \gls{fcc} are governed by the macroscopic circuit topology and the frequency generators $R(\omega)$. 
Because the Pearson correlation normalises variance, the underlying structural dependencies between frequencies remain intact. Consequently, pulse parameters successfully expand the local optimisation manifold (improving trainability) without altering the global structural correlations dictated by the encoding strategy.

\begin{remark}
    To build geometric intuition, one can imagine the coefficient manifold as a 3D object and the macroscopic structural correlations (\gls{fcc}) as its 2D shadow. 
    Introducing pulse parameters adds "texture" to the surface of the object, which increases the variance at the edges of the shadow. 
    
    However, because the macroscopic shape of the object (dictated by the circuit topology) remains intact, the statistical correlation between different regions of the shadow remains identical.
\end{remark}

\section{Conclusion \& discussion}
\label{sec:conclusion}

In this work we presented theoretical insights as to why pulse parameters break the rigid monomial structure of Fourier coefficients and can enlarge the local tangent space of the coefficient map; equivalently, it can increase the rank of the (real) coefficient Jacobian.
This makes an optimisation including the pulse parameters more favourable.
Subsequently we showed that these findings are also confirmed through numerical experiments where we trained a \gls{qfm} on a Fourier series with the same frequencies as the model.
Here we observe that access to the pulse parameters consistently improved the training results throughout all \ase investigated in this work.

While expanding the optimisation space to the pulse-level demonstrates clear theoretical and numerical advantages for trainability, applying these methods to real hardware presents specific trade-offs. 
On one hand, optimising pulse parameters provides the potential to natively mitigate hardware-specific errors, such as crosstalk or decoherence, by shaping pulses to avoid noise channels which is a capability gate-level compilation lacks. 
On the other hand, the number of additional parameters increases significantly depending on the specific gates used and their decomposition (\cf \autoref{tab:num-of-pp-per-gate}).
Thus, a further interesting avenue would be to explore the connection between the effects of noise on \gls{qfm}~\cite{fontana_spectral_2022,franz_out_2025} and the mitigation of these effects enabled by pulse-level optimisation.
As the complexity of training a neural network (and thus a \gls{qfm}) is linear in the number of parameters~\cite{heaton_ian_2018}, the additional complexity added by including pulse parameters is asymptotically negligible for complex models.
However, we acknowledge that the increase of parameters requires to consider trade-offs~\cite{thelen_predict_2025} for quantum computing in general and in the \gls{nisq} era, specifically as overhead required by obtaining gradients through the \gls{psr}~\cite{mitarai_quantum_2018,schuld_evaluating_2019,wierichs_general_2022,periyasamy_guided_2024} is significant. 

A common concern when introducing additional parameters to QML models is the exacerbation of \glspl{bp}, which typically onset when an \as becomes too expressive and approaches a $2$-design. 
However, as noted in~\autoref{subsec:effects-on-fcc-and-expressibility}, tuning pulse parameters does not fundamentally alter the global expressibility (strict bounds of the unitary volume) dictated by the circuit's topological structure and encoding strategy~\cite{gogeissl_quantum_2024}. 
Therefore, extending optimisation to the pulse level is not expected to induce \glspl{bp} any more than a gate-based architecture, while simultaneously smoothing the local landscape to prevent trapping.
Furthermore, while the asymptotic scaling of \glspl{bp} remains bound by the circuit architecture, the localised gradient magnitudes also depend on the derivative of the pulse-envelope function.
However, implementing pulse-level QML on physical QPUs requires extensive low-level access and relies on hardware-in-the-loop gradient estimation techniques (\eg, pulse-level \glspl{psr}), which incur a higher shot-cost overhead. 
Future work will focus on deploying these models on noisy quantum hardware to empirically evaluate the trade-off between enhanced trainability and operational overhead.
While the scaling of training with additional parameters is asymptotically linear, the additional overhead for calculating gradients especially via the \gls{psr} is non-trivial.
Furthermore, the \gls{psr} requires the generator to have two distinct eigenvalues which may not be trivially applicable to arbitrary pulse parameters~\cite{leng_differentiable_2022}. 
In addition to trainability challenges of pulse-level optimisation, we acknowledge that there are physical constraints on the parameters of the pulse (\ie increasing the amplitude $A$ too much can cause leakage) and various error channels in general that are simplified in our work.
While we're optimistic that these challenges can be overcome in practice, we leave an evaluation inspired by the theoretical and numerical findings shown in this manuscript to future work.

\section*{Acknowledgements}
We thank Gabriel Mejía Ruiz for valuable discussions and comments to this work.
MS, EK and AS acknowledge support by the state of Baden-W\"urttemberg through bwHPC.
LS acknowledges the support by the Doctoral School Karlsruhe School of Elementary and Astroparticle Physics: Science and Technology.
MF and WM acknowledge partial support by the German Federal Ministry of Research, Technology and Space (BMFTR), funding program ‘Research Program Quantum Systems’, grant number 13N17387, and by the German Research Foundation, grant MA 9739/1-1. WM acknowledges support by the High-Tech Agenda of the Free State of Bavaria.

\glsresetall
\appendices

\section{Basis Gate Sets}
\label{app:basis_gates}

In practice, quantum hardware only supports a limited set of natively executable operations, known as the basis gate set. 
Any quantum algorithm must be transpiled into a sequence of these basis gates. 
In our experiments, we consider the basis gate set $\{\text{CZ}, \text{RY}, \text{RZ}, \text{RX}\}$. 
This set satisfies the criteria of universal quantum computation, meaning any unitary operation can be approximated to arbitrary precision.
While~\autoref{tab:num-of-pp-per-gate} provides an overview of the parameters required for each of the gates utilised in the \ase of this work, we refer to \citetitle{franz_software_2026}~\cite{franz_software_2026} for more details on the implementation.

\begin{table}[htb]
    \centering
    \caption{Number of pulse parameters for each gate, including the time parameter. Leaf gates (i.e., hardware-native basis gates requiring no further decomposition) are marked with --.}
    \label{tab:num-of-pp-per-gate}
    \begin{tblr}{
            width=\columnwidth,
            colspec={X[c,1.5cm] X[c,1.5cm] X[c]},
            colsep=4pt,
            rowsep=1pt,
        }
        \toprule
        \textbf{Gate}      & \textbf{$\lvert \bPi \rvert$} & \textbf{Decomposition}                                 \\
        \midrule
        \RZ, \CZ       & 1                & --                                            \\
        \RX, \RY       & 3                & --                                            \\
        \midrule
        \hadamard & 4                & $\RZ\cdot\RY$                                 \\
        \Rot      & 5                & $\RZ\cdot\RY\cdot\RZ$                         \\
        \CX       & 9                & $\hadamard\cdot\CZ\cdot\hadamard$             \\
        \CY       & 11               & $\RZ\cdot\CX\cdot\RZ$                         \\
        \CRZ      & 20               & $\RZ\cdot\CX\cdot\RZ\cdot\CX$                 \\
        \CRY      & 24               & $\RY\cdot\CX\cdot\RY\cdot\CX$                 \\
        \CRX      & 26               & $\RZ\cdot\RY\cdot\CX\cdot\RY\cdot\CX\cdot\RZ$ \\
        \bottomrule
    \end{tblr}
\end{table}

\section{Expressibility Metric}
\label{app:expressibility-metric}

The expressibility is evaluated as the \gls{kl}-divergence between the fidelity distributions generated by sampling the Haar integral $\int_{\text {Haar }}\text{d} \psi (\ket{\psi}\bra{\psi})^{\otimes t} $ of a state $t$-design, a given model $\int_{\btheta}\text{d}\btheta \left(\ket{\psi_{\btheta}}\bra{\psi_{\btheta}}\right)^{\otimes t}$ respectively~\cite{sim_expressibility_2019,kullback_information_1951}:
\begin{equation}
    D_{\mathrm{KL}}\left(P_{\text{Model}}(F, \btheta) \| P_{\text {Haar }}(F)\right).
    \label{eq:kl_divergence}
\end{equation}
Here $F \coloneqq \lvert\braket{\psi_1}{\psi_2}\rvert^2$ is the fidelity defined as the squared overlap of a pair of sampled states $(\ket{\psi_1},\ket{\psi_2})$, and $P(F)$ denotes the resulting distribution of fidelities.
When the two distributions coincide, $P_{\text{Model}}(F,\btheta)=P_{\text{Haar}}(F)$, the metric takes the value zero, indicating that the states produced by the \gls{qfm} are exactly Haar-distributed.

\section{Fourier Coefficient Calculation}
\label{app:fourier-coefficient-calculation}

We utilise the definition from Ref.~\cite{strobl_fourier_2025} for calculating the \gls{fcc} as the average over the correlations of coefficients $r_\Theta(\bomega, \bomega')$ between all frequency pairs parametrised by $\btheta \in \Theta$
\begin{equation}
    \fcc \coloneq \frac{1}{\vert \bOmega \vert} \sum_{\bomega, \bomega' \in \bOmega} \left\vert r_\Theta(\bomega, \bomega') \right\vert.
    \label{eq:fcc}
\end{equation}
obtained through the Pearson correlation coefficient~\cite{freedman_statistics_2007} defined as
\begin{equation}
    r_{\Theta}(\bomega, \bomega')= \frac{\sum_{\btheta\in\Theta}\left(c_{\bomega}(\btheta)-\bar c_{\bomega}\right)\left(c_{\bomega'}(\btheta)-\bar c_{\bomega'}\right)}{\sqrt{\sum_{\btheta\in\Theta}\left(c_{\bomega}(\btheta)-\bar c_{\bomega}\right)^{2} \sum_{\btheta\in\Theta}\left(c_{\bomega'}(\btheta)-\bar c_{\bomega'}\right)^{2}}}
    \label{eq:pearson_correlation}
\end{equation}
In the expression above, $\bar c_{\bomega}$ and $\bar c_{\bomega'}$ are the mean values of the coefficients at frequencies $\bomega$ and $\bomega'$, respectively.
The Pearson coefficient is normalised by the standard deviation of the coefficients at each frequency.
When the average correlation across frequency pairs is low, the metric approaches zero; it approaches one for highly correlated \ase.

\section{Fidelity and Trace Distance}
\label{app:fidelity-and-trace-distance}

This section highlights the sensitivity of the fidelity and trace distance to the pulse parameter variance.
Specifically, we evaluate these two metrics with reference to states produced by gate-level simulation of equivalent circuits.
The pulse parameters are then distorted by a scaler which is sampled from a Gaussian distribution of mean $1$ and the variance $\sigma^2_{\blambda}$ being increased incrementally.
\autoref{fig:study-2} shows the result of this evaluation of a pulse parameter variance over $8$ steps within a range from $0$ to $0.008$.
It becomes clear that even slight deviations in $\blambda$ cause the resulting states to diverge rapidly from the target states which is reflected by the fidelity and trace distance going down and up respectively.

\begin{figure}[htb]
    \includegraphics[width=\columnwidth]{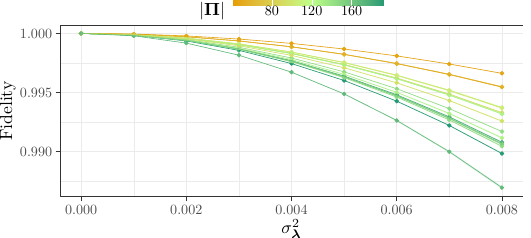}
    \caption{Fidelity between gate- and pulse-level simulations for different levels of distortion $\sigma^2_{\blambda}$, coloured by the number of pulse parameters $\vert \bPi\vert$.}
    \label{fig:study-2}
\end{figure}

\section{Additional Results}
\label{app:additional-results}

\begin{figure}[htb]
    \includegraphics[width=\columnwidth]{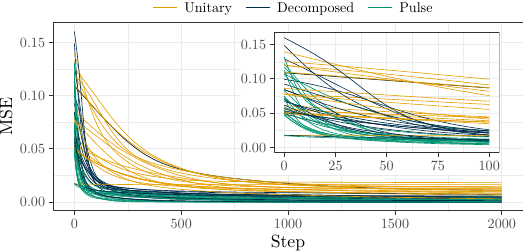}
    \caption{MSE over the course of training for different \ase and using either standard or decomposed unitary gates, or pulse gates. The inset shows a larger view of the first 100 steps.}
    \label{fig:mse-over-training}
\end{figure}

In a similar manner as described in~\autoref{subsec:coefficient-variance}, we evaluated the \gls{fcc} and the expressibility metric.
Here we've found that both metrics do not change significantly \wrt the pulse parameter variance, as shown in~\autoref{fig:expr-fcc-over-distortion} for the \gls{fcc} and the expressibility metric, averaged over ten seeds.

\begin{figure}[htb]
    \includegraphics[width=\columnwidth]{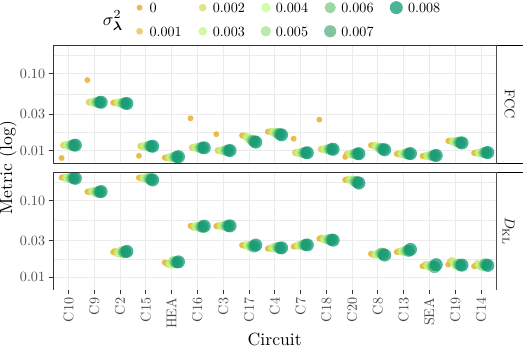}
    \caption{\gls{fcc} and \gls{kl}-divergence $D_\text{KL}$ (\autoref{eq:kl_divergence}, \ie, inverse of the expressibility) over pulse parameter variance $\sigma^2_{\blambda}$ (coloured) for different \ase.}
    \label{fig:expr-fcc-over-distortion}
\end{figure}

\printbibliography

\end{document}